\documentclass[11pt,letterpaper]{article}
\usepackage[english]{babel}
\usepackage[margin=1in]{geometry}
\usepackage{caption}
\usepackage{subcaption}

\usepackage{amsmath}
\usepackage{amsthm}
\usepackage{amssymb}
 \usepackage{amsfonts}
\usepackage{bbm}
\usepackage{booktabs}
\usepackage{xcolor}
\usepackage{graphicx}

\usepackage[utf8]{inputenc}
\usepackage{xifthen,enumitem,comment}
\setlist[enumerate,1]{label=\normalfont{(\Roman*)},leftmargin=2em}
\usepackage{xcolor}
\usepackage{etoolbox}
\makeatletter
\patchcmd{\env@cases}{1.2}{0.96}{}{}
\makeatother
\usepackage{algorithm,algpseudocode}
\usepackage{caption}
\usepackage{subcaption}
\usepackage[hypertexnames=false]{hyperref}
\hypersetup{
  colorlinks=true,
  linkcolor={blue!66!black},
  linkbordercolor=red,
  citecolor={blue!66!black},
}
 \usepackage{cleveref}
 \hypersetup{
 	colorlinks = true,
 	citecolor={blue!66!black},
 }
\usepackage{booktabs,tikz}
\usetikzlibrary{matrix,decorations.pathreplacing}

\definecolor{longhorn}{rgb}{0.8, 0.33, 0.0}
\newcommand{\sw}[1]{\color{longhorn}{#1}\color{black}}

\providecommand{\F}{\mathbbm{F}}
\providecommand{\G}{\mathbbm{G}}

\providecommand{\E}{\mathbbm{E}}
\providecommand{\R}{\mathbbm{R}}
\ifx\argmax\undefined\DeclareMathOperator*{\argmax}{arg\,max}\fi
\ifx\argmin\undefined\DeclareMathOperator*{\argmin}{arg\,min}\fi

\renewcommand{\Pr}{\mathbb{P}}
\renewcommand{\P}{\mathbb{P}}

\newcommand{\of}[1]{\left(#1\right)}
\newcommand{\off}[1]{\left[#1\right]}

\newcommand{\Rev}{\mathrm{Rev}}

\newcommand{\opt}{\mathsf{OPT}}

\newcommand{\ind}{\mathbf{1}}
\newcommand{\dd}[1]{\;\textrm{d}#1}

\newcommand{\Ex}[2]{\underset{{#1}}{\E}\off{{#2}}}
\renewcommand{\Pr}[2][]{\underset{{#1}}{\P}\off{#2}}

\newcommand{\hide}[1]{}
 
\newtheorem{proposition}{Proposition}[section]

\usepackage[hypertexnames=false]{hyperref}
\hypersetup{
  colorlinks=true,
  linkcolor={blue!66!black},
  linkbordercolor=red,
  citecolor={blue!66!black},
}
 \usepackage{cleveref}
 \hypersetup{
 	colorlinks = true,
 	citecolor={blue!66!black},
 }

\usepackage{nicefrac}

\newtheorem{theorem}{Theorem}[section]
\newtheorem{lemma}{Lemma}[section]

\title{Optimal Pricing with Unreliable Signals}
\author{
Zhihao Gavin Tang \thanks{Key Laboratory of Interdisciplinary Research of Computation and Economics, Shanghai University of Finance and Economics \texttt{tang.zhihao@mail.shufe.edu.cn}}
\and
Yixin Tao \thanks{Key Laboratory of Interdisciplinary Research of Computation and Economics, Shanghai University of Finance and Economics \texttt{taoyixin@mail.shufe.edu.cn}}
\and
Shixin Wang \thanks{H. Milton Stewart School of Industrial and Systems Engineering, Georgia Institute of Technology \texttt{shixin.wang@isye.gatech.edu}} 
}
\date{}

\begin{document}
\maketitle
\begin{abstract}
 We study a single-buyer pricing problem with unreliable side information, motivated by the increasing use of AI-assisted decision-making and LLM-based predictions. The seller observes a private sample that may be either accurate (coinciding with the buyer’s valuation), or hallucinatory (an independent draw from the prior), without knowing which case has realized. The buyer does not observe the realized signal, yet knows whether it is accurate or hallucinatory. This creates a higher-order informational asymmetry: the seller is uncertain about the reliability of his own side information, while the buyer has private information about that reliability. 
    
    Adopting a consistency-robustness framework, we characterize the exact Pareto frontier of tradeoffs between consistency (performance under an accurate signal) and robustness (performance under a hallucinatory signal). We show that keeping the unreliable signal private generates substantial value, yielding tradeoffs that strictly dominate any public-signal benchmark. 
    We further show that perfect consistency does not preclude meaningful protection against hallucination: for every prior, there exists a mechanism achieving perfect consistency together with a nontrivial robustness guarantee of $\frac{1}{2}$. Moreover, if the prior has an infinite mean or a mean of at most its monopoly price, we provide a mechanism that is simultaneously 1-consistent and 1-robust. Our results illustrate a new mechanism design paradigm: rather than relying only on information directly possessed by the designer, mechanisms can be built to leverage the other side’s knowledge about the reliability of the designer’s information. 
    
\end{abstract}
\section{Introduction}

We consider the problem of revenue maximization for selling a single item to a single buyer. In the standard Bayesian setting, the buyer privately knows her valuation, while the seller knows only that it is drawn from a commonly known distribution. It is well known that the seller’s optimal mechanism in this case is a take-it-or-leave-it posted price~\cite{mor/Myerson81}.

In modern online marketplaces, the seller often possesses private side information about the buyer, such as her IP address, browsing behavior, or transaction history. 
Such information is naturally modeled as a signal jointly distributed with the buyer’s valuation, and access to such a signal can significantly enhance the seller’s revenue. The seller may use the signal either to refine his posterior belief about the buyer’s valuation and implement price discrimination, or to exploit the buyer’s uncertainty about the signal, which in turn may allow full surplus extraction~\cite{cremer1988full}.

Recent advances in AI, particularly in large language models, have further broadened the range of side information that can be exploited. In addition to conventional structured features, unstructured sources such as the interaction history between the buyer and an LLM may also reveal valuable information about the buyer's type, despite being difficult to incorporate into pricing decisions using traditional methods. Consequently, signals about the buyer's valuation are increasingly generated or refined by machine learning systems.
At the same time, this trend introduces a fundamental challenge: such AI-generated signals may be highly informative, but they may also be unreliable. An LLM may predict the buyer’s valuation accurately when sufficient contextual information is available, but it may also hallucinate and produce a seemingly strong signal even when it has little genuine information about the buyer. 

A similar concern arises more broadly whenever the seller relies on information provided by a third party or intermediary. In many markets, sellers obtain assessments of buyers from outside sources, such as data vendors, brokers, recommendation systems, or other intermediaries that aggregate and process information on their behalf. While such assessments can sometimes be informative about the buyer’s valuation, the seller may not know how they are generated or how reliable they are. In particular, the seller typically does not observe the intermediary’s information set or data-processing effort, and they may see a signal without knowing how informative it truly is about the buyer’s valuation.
If the intermediary exerts sufficient effort, the resulting signal may be genuinely informative; otherwise, it may be little more than a generic guess drawn from the overall population and thus largely uninformative. 
These examples raise a common tension: 
\begin{quote}
How should the seller make use of side information whose informativeness is uncertain?
\end{quote}

A recent work of Lobel, Moreira, and Mouchtaki~\cite{lobel2025} also studies this form of unreliability motivated by machine-learned predictions. They consider a model in which the signal may be either accurate or hallucinatory, and take a Bayesian approach by assuming that the probability of each regime is known. In contrast, we seek a prior-independent understanding of how to leverage unreliable signals without knowing their informativeness.

\subsection{Our Model}

To capture the tradeoff between leveraging informative signals and guarding against unreliable ones, we adopt the consistency-robustness framework~\cite{jacm/LykourisV21} to study revenue maximization with unreliable side information.

Let $\F$ be a commonly known distribution over buyer valuations. The buyer has a private valuation $v \sim \F$, while the seller observes a private signal $s$. 
To capture, in a stylized way, both the power of LLM-generated information and its potential unreliability, we consider the following benchmark model. The signal $s$ is in one of two possible states:
\begin{itemize}
    \item it is \emph{accurate}, namely \(s=v\), in which case it fully reveals the buyer’s valuation. We evaluate a mechanism in this case by its \emph{consistency}: it is \(C\)-consistent if it obtains revenue at least \(C\cdot v\). Here \(v\) is the natural benchmark, since if the seller knew that the signal was accurate, he could extract the full valuation; or
    \item it is \emph{hallucinatory}, namely \(s\) is an independent sample drawn from \(\F\), in which case it provides no information about the buyer’s valuation. We evaluate a mechanism in this case by its \emph{robustness}: it is \(R\)-robust if it obtains revenue at least \(R\cdot \opt(\F)\), where \(\opt(\F)\) denotes the optimal revenue without useful side information. Indeed, if the seller knew that the signal was hallucinatory, he should simply ignore it and post the optimal price for \(\F\).
\end{itemize}

There is an intrinsic tradeoff between consistency and robustness due to the unreliability of the signal. Here \(C,R \in [0,1]\) are approximation factors, with larger values being more desirable. The seller’s goal is therefore to design mechanisms whose achievable pairs \((C,R)\) lie on the Pareto frontier, without knowing whether the observed signal is accurate or hallucinatory.

\paragraph{Public signal.}
A simple baseline for leveraging the unreliable signal \(s\) is to randomize between posting the signal \(s\) itself and posting the monopoly price for \(\F\). This randomized pricing rule achieves every tradeoff \((C,R)\) satisfying \(C+R\le 1\). Our main question is whether the seller can use the unreliable signal in a more sophisticated way and obtain a strictly better tradeoff.

A key observation is that this baseline does not truly benefit from keeping the signal private. Indeed, by posting the price \(s\) with positive probability, the seller effectively reveals the realized signal to the buyer. As a result, the mechanism gives up the seller’s informational advantage and cannot exploit the buyer’s uncertainty about the signal. 

\paragraph{Private signal.}
This motivates our focus on the private-signal setting. At first glance, it is not obvious why privacy should help. If the signal is known to be accurate, then privacy is immaterial: the seller can simply post \(s=v\) and extract the full valuation. If the signal is known to be hallucinatory, then privacy again appears irrelevant, since the signal conveys no information about the buyer’s type. Thus, in either extreme regime considered in isolation, making the signal public seems harmless. The interesting phenomenon arises precisely because the seller does not know which regime has realized. In this case, keeping the signal private turns out to be essential for improving upon the trivial public-signal frontier \(C+R\le 1\).

To make the buyer’s strategic reasoning about the seller’s private information well defined, we impose one final assumption: the buyer knows whether the signal is accurate or hallucinatory. More precisely, when the buyer knows that the signal is accurate, she infers that \(s=v\); when she knows that the signal is hallucinatory, she interprets \(s\) as an independent sample drawn from \(\F\).

This assumption is practically justified by the variance in users' privacy practices and their awareness of their own digital footprints. In LLM-driven marketplaces, a buyer generally knows how much personal information she has revealed. For example, a buyer who opts into data collection, accepts tracking cookies, or provides rich contextual information during LLM interactions to elicit better responses can reasonably infer that the system has captured accurate information about her preferences. Conversely, a privacy-conscious buyer who limits data sharing may infer that any prediction is based only on generic population data.
In broader settings, our assumption is natural in environments where the buyer is better informed than the seller about the provenance of the seller’s side information. The seller may observe a private score, recommendation, or assessment, yet be unable to determine whether it is genuinely based on buyer-specific information or instead generated from a noisy match, incomplete records, or a generic prediction procedure. The buyer often knows whether she used an identified or anonymous account, generated the relevant data, completed the relevant verification process, or directly interacted with the intermediary. Hence, it is plausible that the buyer knows whether the seller’s signal is reliable.

More generally, we only assume that the buyer knows the reliability regime of the signal, not the signal itself. 
From a modeling standpoint, this assumption also serves as a clean benchmark: it isolates the effect of signal privacy by removing the buyer’s uncertainty about the reliability of the information source.

\paragraph{Mechanism design with informed agents.} 
Broadly speaking, our work fits into a broader agenda of mechanism design with informed agents: designing mechanisms that leverage participants’ superior knowledge by integrating elicitation and optimization.
The key informational structure is not merely that the designer lacks information, but that the designer knows that the participant is better informed.
Our setting features a higher-order informational asymmetry: the buyer has private information about information (namely about the reliability of the seller’s signal), while the seller knows only that the buyer knows it. This creates informational leverage for the seller and gives rise to a form of reflexive robustness, in which decisions are designed not only against uncertainty itself, but also based on what the designer knows others know.
This raises a natural question at the interface of mechanism design and information elicitation: rather than first eliciting information and then optimizing based on it, can a mechanism directly leverage the participant’s superior knowledge within the allocation and pricing rule itself? Our model can be viewed as one concrete instance of this paradigm: the seller is uncertain about the reliability of the signal, while the buyer is better informed about the underlying information environment.

Several existing works share this same spirit. Feng and Hartline~\cite{focs/FengH18} show that prior-independent mechanism design can exhibit a non-trivial revelation gap, suggesting that non-truthful mechanisms may outperform truthful ones by relying on information and computation effectively delegated to the agents. 
In prior-independent pricing from samples, the seller observes only a sample from an unknown distribution, known only to belong to a specified ambiguity set. The goal is to design a mechanism based on this sample so as to maximize the worst-case approximation ratio over all distributions in the set. Feng, Hartline, and Li~\cite{stoc/FengHL21} showed that non-truthful mechanisms can outperform truthful ones by leveraging the buyer’s superior knowledge of the underlying distribution. Tang, Tao, and Wang~\cite{TangTW2026} further made this dependence explicit through hidden-pricing mechanisms tailored to the buyer’s informational advantage. Our model is closely connected to this line of work and admits a similar ambiguity-set interpretation. The key difference is that, instead of assuming regularity conditions on the ambiguity set as in those works, we study a structured family consisting of an arbitrary distribution together with all deterministic distributions.

\subsection{Our Results}
Our main result gives a complete characterization of the Pareto frontier of the consistency-robustness tradeoff \((C,R)\).

\begin{theorem}[Optimal Consistency--Robustness Tradeoff]
\label{thm:frontier}
For each \(C\in[0,1]\), define
\[
R^{\star}(C) \;\triangleq\; \inf_{\beta>0}\Bigl[(1+\beta) - C\,(\beta+\beta^2)\,\ln\bigl(1+\tfrac{1}{\beta}\bigr)\Bigr].
\]
Then \(R^{\star}(C)\) is the optimal (largest) robustness guarantee achievable under \(C\)-consistency:
\begin{itemize}
    \item (Achievability) For every distribution \(\F\) and every \(R\le R^{\star}(C)\), there exists a mechanism that is simultaneously \(C\)-consistent and \(R\)-robust.
    \item (Optimality) For every \(R>R^{\star}(C)\), there exists a distribution \(\F\) such that no mechanism can be simultaneously \(C\)-consistent and \(R\)-robust.
\end{itemize}
\end{theorem}

This frontier strictly dominates the public-signal benchmark \(C+R\le 1\) for every \(C\in(0,1]\). In particular, when the designer aims for a symmetric guarantee with \(C=R\), our optimal tradeoff gives an approximation ratio of approximately \(0.822\), substantially improving over the \(0.5\) ratio attained by the naive scheme that randomizes between posting the signal and posting the monopoly price. We illustrate our frontier $R^{\star}(C)$ and the public-signal benchmark in \Cref{fig:ratio}. We prove Theorem~\ref{thm:frontier} in Section~\ref{sec::main-proof}.

\begin{figure}[htbp] 
    \centering 
    \includegraphics[width=0.5\textwidth]{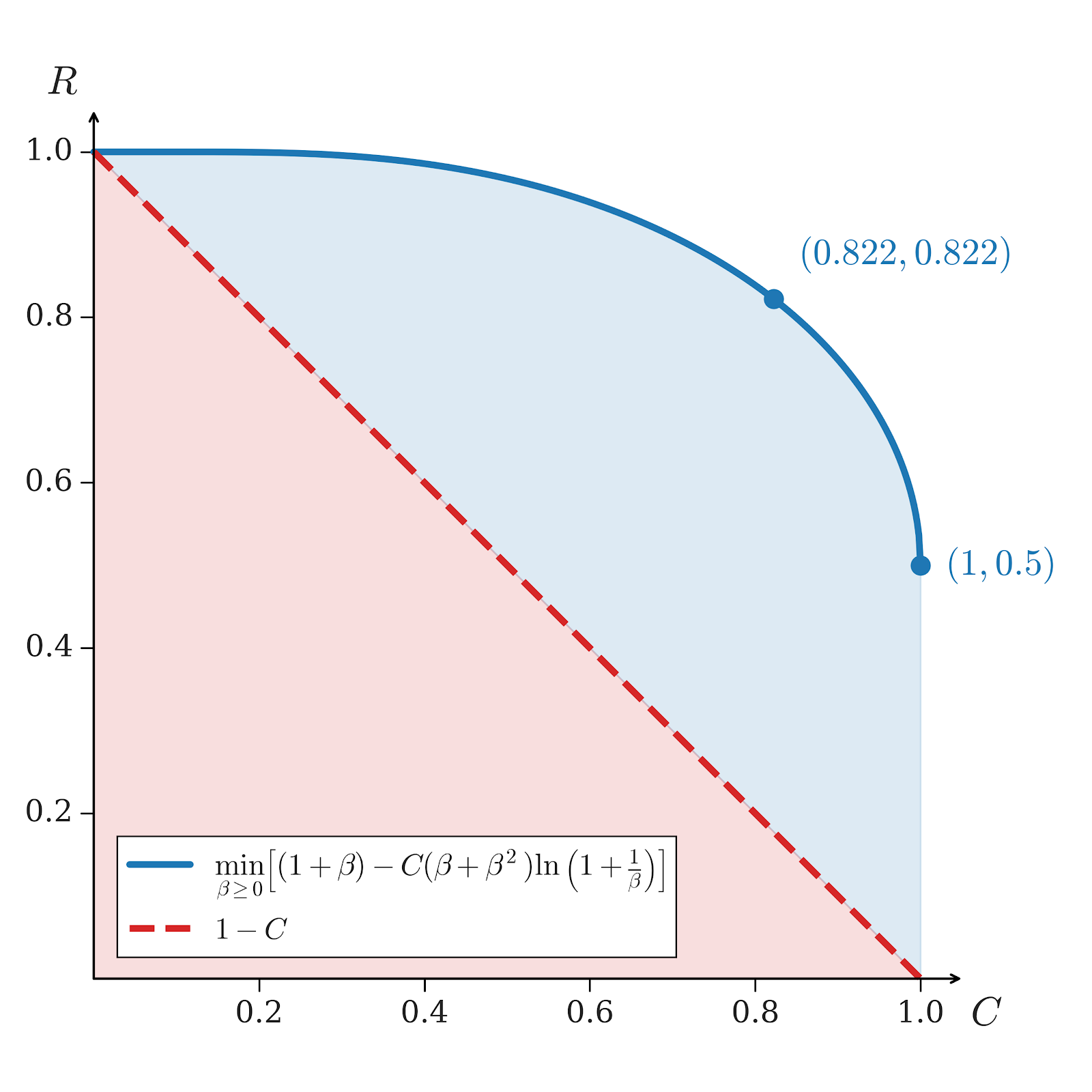} 
    \caption{The blue line represents the $R^{\star}(C)$ curve, which illustrates the robustness-consistency tradeoff frontier achieved in this paper. The dotted red line represents the equation $R + C = 1$, indicating the tradeoff available for the public-signal benchmark.} 
    \label{fig:ratio} 
\end{figure}

\paragraph{Technical Overview.}
Our proof proceeds in two main steps. First, for a fixed prior \(\F\) and consistency requirement \(C\), we characterize the seller's optimal achievable revenue \(\Rev(\F,C)\). Second, we evaluate the worst case performance over all priors \(\F\) to obtain the optimal consistency-robustness tradeoff.

\emph{Step 1: Characterizing Optimal Revenue for a Fixed Prior.} We first reduce the original optimization problem, which is effectively two-and-a-half dimensional in the signal, value, and regime, to a monopolist pricing problem with a nonstandard constraint that places an \emph{upper bound} on the buyer’s utility function. To the best of our knowledge, this is the first work to study mechanism design under such a constraint, as utility upper bounds do not arise naturally outside our setting. By contrast, mechanism design problems with utility \emph{lower bounds} are common, as they naturally capture outside options or opportunity costs.

As detailed in Section~\ref{sec::main-proof}, this reduction begins with a linear program (LP) over general incentive-compatible (IC) and individually rational (IR) mechanisms that condition on both the signal and the reported type. By fixing the mechanism in the accurate regime without harming its performance, we isolate the optimization problem to the hallucinatory regime, condensing it into a single-buyer IC/IR problem governed by the aforementioned utility upper bound. Using a posted-price representation, this becomes an infinite-dimensional LP over a randomized posted price—that is, a distribution $\G$ over prices (see Equation \eqref{eq:lp-ppa-primal}). The resulting optimal mechanism necessarily involves randomization 
and it provides the foundation to evaluate the worst-case approximation ratio.

\emph{Step 2: Establishing the Optimal Consistency-Robustness Tradeoff.} At a high level, our proof is based on a revenue-curve reduction paradigm, but from a different angle than the existing one in the mechanism design literature, e.g., sample-based pricing~\cite{geb/DhangwatnotaiRY15,or/AllouahBB22}, prior-independent auctions~\cite{EC/FuILS15,ms/AllouahB20,focs/HartlineJL20}, simple-versus-optimal mechanisms~\cite{geb/AlaeiHNPY19,sicomp/JinLTX20,stoc/JinLQTX19}. In the classical approach, one typically restricts attention to regular distributions and identifies extremal distributions within that class that minimize the performance guarantee. This allows the original infinite-dimensional problem to be reduced to a finite-dimensional family, whose members are often represented by triangular revenue curves in quantile space. That reduction is deeply tied to regularity. Without such structure, one might expect the worst-case instances to collapse to degenerate distributions, such as simple two-point distributions.

In contrast, we impose no restriction on the distribution family. Instead, we apply linear programming duality to rewrite the revenue $\Rev(\F,C)$ as a minimization over a scalar $\eta$ and a decreasing function $\beta(\cdot)$ (see \eqref{eq:lp-dual}). We then analyze the joint worst-case problem over $\F$ and the dual variables, $\eta$ and $\beta(\cdot)$, (see \eqref{eq:robust-def}). A key structural step reveals that, for any fixed $(\eta,\beta)$ and benchmark $T=\opt(\F)$, the revenue curve of the minimizing prior has a three-threshold “envelope” form (Lemma~\ref{lem:worst_case_dist}).

This structure enables us to further simplify the dual search space: $\eta$ can be eliminated (normalized to zero) without increasing the objective (Lemma~\ref{lem:eliminate_mu}), and it suffices to restrict $\beta$ to step functions (Lemma~\ref{lem:step_beta}). Ultimately, by leveraging the sample through interaction with the buyer, we demonstrate that the extremal distribution takes a triangular form uniquely in value space rather than quantile space. This uncovers a novel extremal geometry for worst-case pricing with unreliable samples.

\paragraph{Perfect Consistency.} An especially striking case is \(C=1\). 
At first sight, perfect consistency appears to leave no room for design: if the signal is accurate, the seller must extract the full value $v=s$, which seems to require pricing directly at the sample. That naive approach, however, yields zero robustness in the hallucinatory case. Our result shows that this intuition is incomplete.
Even under the requirement of perfect consistency, one can still achieve a nontrivial robustness guarantee of \(R=0.5\) for every prior distribution \(\F\), by keeping the realized signal hidden from the buyer and leveraging the buyer’s knowledge of whether the signal is accurate or hallucinatory. 

\begin{theorem}
\label{thm:perfect}
For every distribution \(\F\), there exists a \(1\)-consistent and \(0.5\)-robust mechanism.
\end{theorem}
Although the optimal mechanism under a general consistency constraint can be quite sophisticated, we show that, under the requirement of perfect consistency, a simple guess-for-discount mechanism attains the tight worst-case robustness guarantee of $0.5$.
The mechanism works as follows. The seller first commits to a default price based on the signal, and then offers the buyer an opportunity to guess the realized signal. If the buyer guesses correctly, she may purchase the item at a discounted price; otherwise, she must decide whether to purchase at the default price.

Although our mechanism is primarily conceptual, it highlights a practically relevant design principle: when AI-generated signals are informative but unreliable, it may be preferable not to use them in a fully direct manner, but instead to incorporate them into a pricing rule that is not fully disclosed to the buyer. In this sense, the mechanism can be viewed as a stylized model of AI-assisted pricing with limited disclosure, where the seller uses the signal internally while revealing only coarse or outcome-contingent information. The additional guessing step is mainly a theoretical device for extracting value from the buyer’s knowledge of the signal regime, but it suggests more broadly that carefully designed interaction rules can help sellers benefit from powerful yet imperfect AI predictions.

Beyond these worst-case guarantees over arbitrary distributions, we also identify a broad family of distributions that admits a best-of-both-worlds guarantee, namely, a mechanism that is both perfectly consistent \((C=1)\) and perfectly robust \((R=1)\).

\begin{theorem}
\label{thm:best-of-both}
Let \(\mu(\F)=\E_{s\sim\F}[s]\) denote the mean of \(\F\), and let \(p^*(\F)\) denote its monopoly price. Then the following hold:
\begin{itemize}
    \item If \(\mu(\F)\le p^*(\F)\), then there exists a mechanism that achieves \(C=R=1\).
    \item If \(\mu(\F)=\infty\), $p^*(\F)<\infty$, then for every \(\epsilon>0\), there exists a mechanism such that \(\min(C,R)\ge 1-\epsilon\).
\end{itemize}
\end{theorem}

Natural examples satisfying \(\mu(\F)\le p^*(\F)\) include the uniform distributions whose support starts at $0$ and exponential distributions. For the infinite-mean case, canonical examples include the equal-revenue distribution and other sufficiently heavy-tailed distributions with infinite mean.

Our results for the perfectly consistent regime, including Theorem~\ref{thm:perfect} and Theorem~\ref{thm:best-of-both} are presented in Section~\ref{sec:perfect}.

\subsection{Related Work}

\paragraph{Mechanism Design with Correlated Information.} The standard approach to modeling side information is through a joint distribution over the signal and the buyer’s valuation. A classic result of Cr{\'e}mer and McLean~\cite{cremer1988full} shows that, in multi-bidder auctions, sufficiently correlated and private signals can be leveraged to extract the full surplus when only interim individual rationality is required. Although such full surplus extraction is impossible in general under ex post individual rationality, Fu, Liaw, Lu, and Tang~\cite{soda/FuLLT18} show that the separation between private and public signals remains dramatic: the ratio between their optimal revenues can still be unbounded.

\paragraph{Prior-Independent Mechanism Design with Samples.} Our model can also be interpreted as a prior-independent mechanism design problem with a single sample. This line of work studies settings in which the seller has only very limited knowledge of the valuation distribution. In particular, the seller is given a single sample---either public~\cite{geb/DhangwatnotaiRY15,SIAM/HuangMR18,EC/FuILS15,or/AllouahBB22} or private~\cite{stoc/FengHL21,TangTW2026}---from the underlying distribution, and aims to design approximately optimal mechanisms, often for classes such as \(\alpha\)-regular distributions. Follow-up works have further explored the power of two samples~\cite{EC/BabaioffGMM18,EC/DaskalakisZ20}.

Our work is also related to Fu et al.~\cite{jet/FuHHK21}, who study sample-based surplus extraction under ambiguity about the ground-truth distribution. In contrast to our setting, their mechanism is DSIC, the samples are public, and the goal is full surplus extraction simultaneously for all candidate distributions in the ambiguity set.

\paragraph{Mechanism Design with Predictions.}
Our work contributes to the rapidly growing literature on algorithms with predictions, also known as learning-augmented algorithms. The consistency-robustness framework of Lykouris and Vassilvitskii~\cite{jacm/LykourisV21} has enabled a refined analysis of data structures, algorithms, and mechanisms; see \cite{ALPS} for an up-to-date list of papers. We focus here on the branch most closely related to our setting, namely mechanism design with predictions; see \cite{sigecom/BalkanskiGT23} for an annotated reading list.
Mechanism design with predictions was initiated by Agrawal et al.~\cite{mor/AgrawalBGOT24} and Xu and Lu~\cite{ijcai/XuL22}. Since then, it has been studied in a variety of settings, including facility location~\cite{mor/AgrawalBGOT24,corr/IstrateB22,corr/BalkanskiDGT25}, strategic scheduling~\cite{itcs/BalkanskiGT23}, auctions~\cite{nips/MedinaV17,ijcai/XuL22,nips/PrasadBS23,ijcai/CaragiannisK24}, competitive auctions~\cite{ec/LuWZ24}, online auctions~\cite{ec/BalkanskiGTZ24}, clock auctions~\cite{soda/GkatzelisST25}, and auto-bidding~\cite{www/AggarwalGTZ25}.

To the best of our knowledge, our work is the first to study value predictions in a Bayesian setting from a prior-independent perspective. We argue that independence between the prediction and the buyer’s valuation is the natural adversarial benchmark in this framework: any correlation between the prediction and the ground truth would already provide useful information to the mechanism designer, while an independent prediction represents the case in which the signal carries no information about the buyer’s type.
  
\section{Preliminary}
We study the mechanism design problem in which a seller aims to sell one product to a single buyer. The buyer has a private valuation $v \in [0, v_{\max}]$ for the product, which is realized in one of two ways: it is either an arbitrary value or drawn from a prior distribution $\F : [0, v_{\max}] \rightarrow [0, 1]$, where we use the convention $\F(v) \triangleq \P[V < v]$. The seller knows $\F$, but does not know the buyer's valuation $v$, nor whether $v$ is arbitrary or sampled from $\F$. The seller has access to a \emph{private} sample $s$. If the buyer's value is arbitrary, this sample perfectly correlates with the value ($s = v$); if the buyer's value is drawn from $\F$, the sample is drawn independently from the same distribution ($s \sim \F$). The buyer is fully aware of $\F$, her true valuation $v$, and its realization process, but she does not know the seller's observed sample $s$.

We study revenue-maximizing mechanisms that ask the buyer to report both her private valuation and its realization process. A mechanism consists of an allocation rule $x: \mathbb{R} \times \mathbb{R} \times \{\textsf{A}, \textsf{H}\} \to [0,1]$ and a payment rule $p: \mathbb{R} \times \mathbb{R} \times \{\textsf{A}, \textsf{H}\} \to \mathbb{R}$, which map the seller's sample $s$, the buyer's reported value $v$, and the buyer's reported realization process (either $\textsf{A}$ or $\textsf{H}$, where $\textsf{A}$ represents accurate state, namely $s = v$, and $\textsf{H}$ represents the hallucinatory state, namely $s$ is an independent sample drawn from $\F$) to an allocation probability and a payment, respectively. 
The seller seeks to design a mechanism that achieves the following two guarantees:
\begin{itemize}
    \item Consistency: If the buyer's private valuation is an arbitrary deterministic value, the seller's revenue is at least $C \cdot v$;
    \item Robustness:  If the buyer's private valuation is drawn from a prior distribution $\F$, the seller's revenue is at least $R \cdot \opt(\F)$.
\end{itemize} 

By the revelation principle, we restrict our attention to mechanisms that satisfy incentive compatibility (IC) and individual rationality (IR). Truthful reporting of both private valuation and its realization process must maximize the buyer's expected utility under the seller's sample $s$: 
\begin{align*}
(v,\textsf{H}) \to (v',\textsf{H}): \quad & \Ex{s\sim \F}{v \cdot x(s,v,\textsf{H}) - p(s,v,\textsf{H})} \ge \Ex{s\sim \F}{v \cdot x(s,v',\textsf{H}) - p(s,v',\textsf{H})} \\
(v,\textsf{A}) \to (v',\textsf{H}): \quad & v \cdot x(v,v,\textsf{A}) - p(v,v,\textsf{A}) \ge v \cdot x(v,v',\textsf{H}) - p(v,v',\textsf{H})
\end{align*}
When the buyer's valuation is drawn from $\F$, that is, in the hallucinatory state $\textsf{H}$, it is without loss to consider deviations that misreport only the valuation while continuing to report $\textsf{H}$. A deviation to $\textsf{A}$ would require the buyer to report a value consistent with the seller's realized signal $s$, which is not observed by the buyer. If $\F$ is non-degenerate, then for any fixed report $v'$ we have $\P_{s\sim \F}[s \neq v'] > 0$. Therefore, the seller can condition the mechanism on this mismatch and impose a sufficiently severe penalty, so such deviations are not profitable. The same reasoning rules out deviations of the form $(v,\textsf{A}) \to (v',\textsf{A})$ with $v' \neq v$, since in the accurate state the seller observes $s=v$ and can verify whether the buyer's report is consistent with the realized signal.

The optimal mechanism under the consistency $C$ is then captured by the following optimization: 
\begin{align*}
\Rev(\F,C) ~\triangleq~ \max_{x, p}: \quad \Ex{s\sim \F, v \sim \F}{p(s,v, \textsf{H})} \tag{$\star$} \label{lp:opt}
\end{align*}
subject to the following constraints: 
\begin{align*}
&~~\text{Consistency:}&& p(v,v,\textsf{A}) &&\ge C \cdot v \quad & \forall v \\
&~~ \text{IC:} &&\Ex{s\sim \F}{v \cdot x(s,v,\textsf{H}) - p(s,v,\textsf{H})} &&\ge \Ex{s\sim \F}{v \cdot x(s,v',\textsf{H}) - p(s,v',\textsf{H})}\quad & \forall v,v' \\
&&& v \cdot x(v,v,\textsf{A}) - p(v,v,\textsf{A}) &&\ge v \cdot x(v,v',\textsf{H}) - p(v,v',\textsf{H})\quad  & \forall v,v' \\
&~~\text{IR:} && \Ex{s\sim \F}{v \cdot x(s,v,\textsf{H}) - p(s,v,\textsf{H})} &&\ge 0 \quad& \forall v \\
&&& v \cdot x(v,v,\textsf{A}) - p(v,v,\textsf{A}) &&\ge 0 \quad & \forall v \\
&~~\text{Feasibility:}&& 0 \le x(s,v,\textsf{H}) \le 1 ~~ &&\text{and} \quad 0 \le x(v,v,\textsf{A}) \le 1 \quad & \forall s,v. \\
\end{align*}

Additionally, the robustness $R(C)$ is the minimal ratio between the revenue achieved, $\Rev(\F,C)$, and the optimal revenue, $\opt(\F)$, over all $\F$: $$R(C) \triangleq \min_{\F} \frac{\Rev(\F,C)}{\opt(\F)}.$$

\section{Mechanisms with Perfect Consistency}
\label{sec:perfect}

In this section, we focus on mechanisms with perfect consistency. We present several examples illustrating the space of feasible mechanisms. To highlight the main ideas, we describe the mechanisms in an intuitive way and reason from the buyer’s perspective, rather than specifying them in full formal detail and explicitly verifying the IC and IR constraints.

Somewhat surprisingly, even under the requirement of perfect consistency, one can still obtain nontrivial robustness guarantees for every prior distribution \(\F\). Moreover, for certain distributions, it is possible to achieve the best of both worlds, namely \(C=R=1\).

Throughout this section, for ease of presentation, we assume that the distribution \(\F\) is continuous. All of the mechanisms constructed below involve a step in which the buyer is asked to guess the signal. Under the continuity assumption, the buyer guesses correctly with probability zero in the hallucinatory case, which simplifies both the description and the analysis of the mechanisms. By contrast, in the accurate-signal case, the buyer is able to guess correctly, and by construction this is always the most favorable outcome for her.

This technical issue can also be resolved for general distributions with atoms by appropriately modifying the mechanism so that incorrect guesses incur a penalty while the buyer retains the option to abstain from guessing. This corresponds to the observation that the IC constraints associated with deviations of the form \((v,\textsf{H})\to (v',\textsf{A})\) are not relevant for any non-degenerate distribution \(\F\).
We omit the implementation details and note that the statement below remains valid for general distributions.

\subsection{Worst-Case Robustness with Perfect Consistency}
In this section, we prove Theorem~\ref{thm:perfect}.

\begin{proof}[Proof of Theorem~\ref{thm:perfect}]
Let \(p^+(s)\) denote the revenue-maximizing price subject to being at least \(s\), namely
\[
p^+(s)\in \argmax_{p\ge s} p\cdot (1-\F(p)).
\]
Consider the following mechanism. The seller sets a default price \(p^+(s)\) based on the signal \(s\), and keeps this price hidden from the buyer. The buyer is then given an opportunity to guess the realized signal \(s\). After she submits her guess, the seller reveals the price: if the buyer guessed correctly, the seller discounts the price to \(s\); otherwise, the price remains \(p^+(s)\). The buyer then decides whether to purchase by comparing her valuation with the resulting price.

We prove this mechanism is \(1\)-consistent and \(0.5\)-robust.

\paragraph{Accurate-signal case.} If the signal is accurate, then \(s=v\). By guessing the signal correctly, the buyer obtains the item at price \(s=v\), yielding utility \(0\). If she guesses incorrectly, then the price remains \(p^+(v)\ge v\), so she weakly prefers not to buy. Thus, under seller-favoring tie-breaking\footnote{When the buyer is indifferent between buying and not buying, we assume she buys (tie-breaking in favor of the seller).}, the buyer guesses correctly and buys the item, and the seller obtains revenue \(v\). Hence the mechanism is \(1\)-consistent.

\paragraph{Hallucinatory-signal case.} 
In this case, \(s\) is an independent sample from \(\F\). Since \(\F\) is continuous, the buyer guesses the exact value of \(s\) correctly with probability zero. Therefore, the discount is never triggered, and the mechanism is equivalent to posting the hidden price \(p^+(s)\), where \(s\sim \F\) is an independent sample.

Thus the expected revenue in the hallucinatory case is
\[
\E_{s\sim \F}\bigl[p^+(s)\cdot (1-\F(p^+(s)))\bigr].
\]
We show that this quantity is at least \(0.5 \cdot \opt(\F)\).

Let \(r(q) \triangleq q \cdot \F^{-1}(1-q)\) be the revenue curve of \(\F\).
Let \(q^*\) be a maximizer of \(r\), so that 
\[
\opt(\F) = \max_{q\in[0,1]} r(q) = r(q^*)~.
\]
For each \(q\in[0,1]\), define
\[
r^+(q) \triangleq \max_{t\le q} r(t).
\]

If the sample \(s\) has quantile \(q\), then the feasible prices \(p\ge s\) correspond exactly to quantiles \(t\le q\). Therefore, by definition of \(p^+(s)\), the revenue generated by the posted price \(p^+(s)\) is exactly \(r^+(q)\). Since \(\F\) is continuous, the quantile of \(s\) is uniformly distributed on \([0,1]\). Hence the expected revenue of the mechanism in the hallucinatory case is \( \int_0^1 r^+(q)\dd q \). 

Now observe that \(\F^{-1}(1-q)=r(q)/q\) is non-increasing in \(q\). Therefore, for every \(q\le q^*\),
\[
r^+(q) \ge r(q)=q\, \F^{-1}(1-q)\ge q\,\F^{-1}(1-q^*) = \frac{q}{q^*} \opt(\F).
\]
On the other hand, for every \(q\ge q^*\), the quantile \(q^*\) is feasible in the definition of \(r^+(q)\), so
\[
r^+(q)\ge r(q^*) = \opt(\F).
\]
Therefore,
\[
\int_0^1 r^+(q)\dd q \ge \int_0^{q^*} \frac{q}{q^*} \opt(\F)\dd q + \int_{q^*}^1 \opt(\F) \dd q =  \left(1- \frac{q^*}{2}\right) \opt(\F) \ge \frac{1}{2} \cdot \opt(\F).
\]
This proves that the mechanism is \(0.5\)-robust.
\end{proof}
Readers familiar with the pricing-with-samples literature may recall that posting a sample \(s\sim \F\) achieves a \(0.5\)-approximation to \(\opt(\F)\) for every regular distribution \(\F\). By contrast, for general distributions, posting the sample itself can have zero approximation ratio. Here, we show that replacing the sample price with the revenue-maximizing price subject to being at least \(s\) restores the \(0.5\)-approximation, and in fact guarantees it for every distribution \(\F\).

\subsection{Best of Both Worlds}
In this section, we prove Theorem~\ref{thm:best-of-both}.
\paragraph{Warm-up: Uniform \([0,1]\).}
We start with the simple case where \(\F\) is the uniform distribution on \([0,1]\).
Consider the following pricing mechanism: the seller sets the price equal to the unreliable signal, but keeps the realized price private until after the buyer decides whether to purchase. Equivalently, the buyer must decide whether to buy before observing the realized price. The buyer’s reasoning is then as follows:
\begin{itemize}
    \item If the signal is accurate, then she infers that the price equals her valuation, and therefore is willing to buy. This yields \(C=1\).
    \item If the signal is hallucinatory, then she views the price as an independent sample from the uniform distribution on \([0,1]\), and hence evaluates it by its expectation \(\E_{s\sim \F}[s]=1/2\). She therefore chooses to buy if and only if \(v \ge 1/2\). Remarkably, this exactly implements the monopoly price \(1/2\), and thus the mechanism also achieves \(R=1\).
\end{itemize}
The uniform example already illustrates the additional power afforded by signal privacy. In both the accurate and the hallucinatory regimes, the mechanism effectively sells the item at the mean of the signal distribution. The key coincidence in this example is that, for the uniform distribution on \([0,1]\), the mean and the monopoly price are the same, both equal to \(1/2\). As a result, pricing at the mean simultaneously delivers perfect consistency and perfect robustness.

More importantly, this example is not meant to suggest that the seller can only benefit from the coincidence between the mean and the monopoly price. Rather, for more general distributions, we will present mechanisms that make more refined use of the private signal and continue to exploit the additional power provided by signal privacy. In particular, we identify a broad family of distributions that admit mechanisms which are both perfectly consistent and perfectly robust; we refer to this phenomenon as the \emph{best-of-both-worlds} guarantee.

The best-of-both-worlds guarantee for the uniform distribution in fact extends to every distribution \(\F\) whose monopoly price \(p^*(\F)\) is at least its mean \(\mu(\F)\). To achieve this, we generalize the uniform mechanism by introducing an additional guessing step, which enables a more refined use of the private signal while also exploiting the buyer’s knowledge of whether the signal is accurate or hallucinatory.

Below, we prove a slightly more general statement, showing that the ratio between the monopoly price and the mean of the distribution provides a lower bound on the consistency achievable under perfect robustness.

\begin{proposition}
For every distribution \(\F\), there exists a mechanism that is
\[
C=\min\!\left(\frac{p^*(\F)}{\mu(\F)},\,1\right)\text{-consistent}
\qquad\text{and}\qquad
R=1\text{-robust}.
\]
In particular, if \(p^*(\F)\ge \mu(\F)\), then there exists a mechanism that achieves \(C=R=1\).
\end{proposition}

\begin{proof}
Let $C=\min\!\left(\frac{p^*(\F)}{\mu(\F)},\,1\right)$. Consider the following mechanism. The seller sets a default price equal to
\[
C \cdot s + p^*(\F) - C \cdot \mu(\F),
\]
and keeps the realized price private from the buyer. The buyer must decide whether to purchase before learning the realized signal or price. If she chooses to buy, she is then given an opportunity to guess the realized signal \(s\): if her guess is correct, the seller discounts the price to \(C \cdot s\); otherwise, the default price remains unchanged.

The buyer reasons as follows:
\begin{itemize}
    \item If the signal is accurate, then the buyer infers that \(s=v\). Hence, by choosing to buy and correctly reporting \(s\), she obtains the item at price \(C\cdot s= C \cdot v\), yielding non-negative utility. She therefore chooses to buy. Thus, the mechanism is \(C\)-consistent.

    This argument uses the fact that \(p^*(\F)\ge C \cdot \mu(\F)\) by the definition of $C$. Indeed, if \(p^*(\F)< C \cdot \mu(\F)\), then the so-called discount would actually increase the price. In that case, when the signal is accurate, the buyer would prefer to guess incorrectly rather than correctly, and the consistency guarantee would fail.
    \item If the signal is hallucinatory, then \(s\) is an independent sample from \(\F\). Since \(\F\) is continuous and the realized signal is private, the buyer guesses \(s\) correctly with probability zero. Therefore, the discount option has no effect, and she evaluates the hidden price by its expectation:
    \[
    \E\!\left[C \cdot s + p^*(\F) - C \cdot \mu(\F)\right] = p^*(\F).
    \]
    It follows that the mechanism effectively implements the monopoly price, and hence achieves \(R=1\).
\end{itemize}
\end{proof}

The mechanism above has two qualitatively different interpretations, depending on the relation between \(p^*(\F)\) and \(\mu(\F)\). When \(p^*(\F)\le \mu(\F)\), we have \(C=p^*(\F)/\mu(\F)\), and the mechanism reduces to simply posting the hidden price \(C\cdot s\): in the accurate case this yields revenue \(C\cdot v\), while in the hallucinatory case the buyer evaluates the price by its expectation \(C\mu(\F)=p^*(\F)\). Thus, in this regime the mechanism is exactly the scaled analogue of the uniform warm-up.

By contrast, when \(p^*(\F)>\mu(\F)\), we have \(C=1\), so the default price becomes \(s+p^*(\F)-\mu(\F)\), while a correct guess reduces it to \(s\). In this regime, the guessing step is essential for attaining perfect consistency, since it is precisely what allows the buyer to lower the price from the default level back to her true value.

The previous proposition suggests that the main difficulty arises when the mean of the distribution is large. Somewhat surprisingly, however, even in the extreme case where \(\mu(\F)=\infty\), one can still obtain an almost best-of-both-worlds guarantee, albeit via a different mechanism.

\begin{proposition}
If \(\mu(\F)=\infty\) and \(p^*(\F)<\infty\), then for every \(\epsilon>0\), there exists a \(C\)-consistent and \(R\)-robust mechanism with \(\min(C,R)\ge 1-\epsilon\).
\end{proposition}

\begin{proof}
Fix \(\epsilon>0\). Since \(\mu(\F)=\infty\), the truncated expectation
\(m(a):=\E_{v\sim\F}[v\cdot\ind[v\le a]]\) tends to \(+\infty\) as \(a\to\infty\), while
\(q(a) = \P_{v\sim\F}[v<a]\) tends to \(1\). Hence, we may choose \(a\) such that
\[
q(a)>1-\epsilon
\qquad\text{and}\qquad
m(a)>q(a)\cdot p^*(\F).
\]
Fix such an \(a\), and write \(q:=q(a)\).

Next, since \(\E_{v\sim\F}[v\cdot\ind[v>a]]=\infty\), we may choose \(b>a\) such that\footnote{If exact equality is impossible because of an atom at \(b\), one may adjust the transfer at \(b\) in the obvious way.} 
\[
\epsilon\cdot \E_{v\sim\F}[v\cdot\ind[a<v\le b]]
=
m(a)-q\,p^*(\F).
\]

We construct a mechanism with \(C=1-\epsilon\) and \(R=q>1-\epsilon\). The buyer first decides whether to participate. If she participates, she is asked to guess the realized signal \(s\).
\begin{itemize}
    \item If she guesses correctly, she receives the item and pays \((1-\epsilon)s\).
    \item If she guesses incorrectly, then \(s\) is realized, and:
    \begin{itemize}
        \item if \(s\le a\), she receives the item and pays \(s\);
        \item if \(a<s\le b\), she does not receive the item, and the seller pays her \(\epsilon s\);
        \item if \(s>b\), the item remains unsold and no transfer is made.
    \end{itemize}
\end{itemize}

We now analyze the buyer’s behavior.

\paragraph{Accurate-signal case.}
If the signal is accurate, then \(s=v\). By guessing correctly, the buyer obtains utility
\[
v-(1-\epsilon)v=\epsilon v.
\]
If she guesses incorrectly, her utility is 
\[
\begin{cases}
0, & v\le a,\\
\epsilon v, & a<v\le b,\\
0, & v>b.
\end{cases}
\]
Thus, guessing correctly is weakly better than guessing incorrectly. Under seller-favoring tie-breaking, she therefore participates and guesses correctly. Hence the mechanism is \(C\)-consistent with \(C=1-\epsilon\).

\paragraph{Hallucinatory-signal case.}

If the signal is hallucinatory, then \(s\) is an independent sample from \(\F\). Since \(\F\) is continuous, the buyer guesses the exact value of \(s\) correctly with probability zero. Therefore, conditional on participating, she almost surely ends up in the wrong-guess branch, and only needs to decide whether participating is worthwhile.

Her expected utility from participating is
\[
q\cdot v
-\E_{s\sim \F}[s\cdot\ind[s\le a]]
+\epsilon\cdot\E_{s\sim \F}[s\cdot\ind[a<s\le b]].
\]
By the choice of \(b\), this equals
\[
q\cdot v - m(a) + \bigl(m(a)-q\,p^*(\F)\bigr)
=
q\cdot (v-p^*(\F)).
\]
Therefore, the buyer participates if and only if \(v\ge p^*(\F)\).

Whenever she participates, the seller’s expected revenue is
\[
\E_{s\sim \F}[s\cdot\ind[s\le a]]
-\epsilon\cdot\E_{s\sim \F}[s\cdot\ind[a<s\le b]]
=
m(a)-\bigl(m(a)-q\,p^*(\F)\bigr)
=
q\,p^*(\F).
\]
Hence, in the hallucinatory case, the seller obtains expected revenue
\[
q\,p^*(\F)\cdot \P_{v\sim\F}[v\ge p^*(\F)]
=
q\cdot \opt(\F).
\]
Thus the mechanism is \(R\)-robust with \(R=q>1-\epsilon\), which completes the proof.
\end{proof}

The role of the heavy tail is to provide enough expected mass in rare large signals to calibrate the buyer’s utility in the hallucinatory case. More precisely, after fixing \(a\), the mechanism uses the interval \((a,b]\) together with the transfer \(\epsilon s\) to exactly offset the expected payment from the low-signal region \([0,a]\), making the buyer’s participation utility proportional to \(v-p^*(\F)\). Since \(\mu(\F)=\infty\), one can always choose \(b>a\) so that this calibration is possible. This is what allows the mechanism to capture a \(q(a)\)-fraction of the optimal revenue. Because \(q(a)\) can be made arbitrarily close to \(1\), the revenue can be made arbitrarily close to optimal.

\section{Optimal Consistency--Robustness Tradeoff (Theorem~\ref{thm:frontier})} \label{sec::main-proof}
In this section, we prove Theorem~\ref{thm:frontier}. In Section~\ref{subsec::main-proof-1}, we characterize the seller's optimal achievable revenue \(\Rev(\F,C)\), for a fixed prior \(\F\) and a given consistency requirement \(C\), and, in Section~\ref{subsec::main-proof-2}, we analyze the robustness guarantee over all \(\F\).

\subsection{Optimal Mechanism under the Consistency Constraint} \label{subsec::main-proof-1}
In this section, we simplify the linear program \eqref{lp:opt}. In particular, we show that the original formulation, whose decision variables are $x(\cdot,\cdot,\textsf{A}/\textsf{H})$ and $p(\cdot,\cdot,\textsf{A}/\textsf{H})$, can be reduced to an equivalent program involving only two one-dimensional functions, $x(\cdot)$ and $p(\cdot)$.

We first observe that in \eqref{lp:opt}, $x(\cdot, \cdot, \textsf{A})$ and $p(\cdot, \cdot, \textsf{A})$ appear exclusively within the following constraints:
\begin{align*}
    v \cdot x(v,v,\textsf{A}) - p(v,v,\textsf{A}) &\ge v \cdot x(v,v',\textsf{H}) - p(v,v',\textsf{H})\quad  & \forall v,v' \\
    v \cdot x(v,v,\textsf{A}) - p(v,v,\textsf{A}) &\ge 0 \quad & \forall v \\
    x(v,v,\textsf{A}) &\leq 1 \quad & \forall v \\
    p(v,v,\textsf{A}) &\geq C \cdot v \quad & \forall v
\end{align*}
Setting $x(v, v, \textsf{A}) = 1$ and $p(v, v, \textsf{A}) = C \cdot v$ for every $v$ will not decrease the optimal value of \eqref{lp:opt}. Applying these optimal choices reduces \eqref{lp:opt} to the following equivalent linear program, where $x(s,v,\textsf{H})$ is replaced by $x(s,v)$ and $p(s,v,\textsf{H})$ is replaced by $p(s,v)$:
\begin{subequations}
\addtocounter{equation}{-1}
\begin{align}
\max_{x,p}: \quad & \Ex{s,v\sim \F}{p(s,v)} \label{eq:lp-obj1}\\
\text{subject to}: \quad & \Ex{s\sim \F}{v \cdot x(s,v) - p(s,v)} \ge \Ex{s\sim \F}{v \cdot x(s,v') - p(s,v')} & \forall v,v' \label{eq:lp-a1}\\
& \Ex{s \sim \F}{v \cdot x(s,v) - p(s,v)} \ge 0 & \forall v \label{eq:lp-b1}\\
& (1-C) \cdot s \ge s \cdot x(s,v) - p(s,v) & \forall s,v \label{eq:lp-c1} \\ 
& 0 \le x(s,v) \le 1 & \forall s,v
\end{align}
\end{subequations}

We show that the above program is equivalent to the following linear program. Hence, for any given consistency level $C$, the optimal robust mechanism is obtained by solving an LP over the one-dimensional functional variables $x(\cdot)$ and $p(\cdot)$. When $\F$ has finite support, this LP is finite-dimensional and can be solved directly.
\begin{subequations}
\addtocounter{equation}{-1}
\begin{align}
\max_{x,p}: \quad & \Ex{v \sim \F}{p(v)} \label{eq:lp-obj2}\\
\text{subject to}: \quad & v \cdot x(v) - p(v) \ge v \cdot x(v') - p(v') & \forall v,v' \label{eq:lp-a2}\\
& v \cdot x(v) - p(v) \ge 0 & \forall v \label{eq:lp-b2}\\
& v \cdot x(v) - p(v) \le (1-C) \cdot \Ex{s \sim \F}{s} + \Ex{s \sim \F}{(v-s)^+} & \forall v \label{eq:lp-c2}\\ 
& 0 \le x(v) \le 1 & \forall v
\end{align}
\end{subequations}
\begin{lemma}
The linear program \eqref{eq:lp-obj2} is equivalent to the linear program \eqref{lp:opt}.
\end{lemma} 
\begin{proof}
It suffices to show that the optimal value of \eqref{eq:lp-obj1} is no greater than that of \eqref{eq:lp-obj2} and vice versa.

First, consider constraint \eqref{eq:lp-c1}. Taking expectation over $s\sim \F$ yields
\[
(1-C) \cdot \Ex{s\sim \F}{s} \ge \Ex{s\sim \F}{s \cdot x(s,v)} - \Ex{s\sim \F}{p(s,v)} = v \cdot \Ex{s\sim \F}{x(s,v)} - \Ex{s\sim \F}{p(s,v)} - \Ex{s\sim \F}{(v-s) \cdot x(s,v)} 
\]
 The objective of \eqref{eq:lp-obj1} becomes larger if we replace constraint \eqref{eq:lp-c1} by the inequality above.
Define $x(v) = \E_{s\sim \F}[x(s,v)]$ and $p(v)= \E_{s\sim \F}[p(s,v)]$. Then the preceding inequality becomes
\begin{equation*}
 v \cdot x(v) - p(v) \le (1-C) \cdot \Ex{s\sim \F}{s} + \Ex{s\sim \F}{(v-s) \cdot x(s,v)} \le (1-C) \cdot \Ex{s\sim \F}{s} + \Ex{s\sim \F}{(v-s)^+}~,
\end{equation*}
where the inequality holds by $x(s,v)\in[0,1]$ for all $s,v$. Thus, constraint \eqref{eq:lp-c1} implies \eqref{eq:lp-c2}. By the construction of $x(v)$ and $p(v)$, constraints \eqref{eq:lp-a1}, \eqref{eq:lp-b1} imply the corresponding constraints \eqref{eq:lp-a2} \eqref{eq:lp-b2}. 
Hence, for any feasible solution in \eqref{eq:lp-obj1}, we can construct a feasible solution of \eqref{eq:lp-obj2} with the same objective value by setting $x(v) = \E_{s\sim \F}[x(s,v)]$ and $p(v)= \E_{s\sim \F}[p(s,v)]$, so the objective of  \eqref{eq:lp-obj1} is no greater than that of \eqref{eq:lp-obj2}.

We now show that any feasible pair $(x(\cdot),p(\cdot))$ for \eqref{eq:lp-obj2} can be implemented by a pair $(x(\cdot,\cdot),p(\cdot,\cdot))$ feasible for \eqref{eq:lp-obj1}. The construction is as follows. For each $v$, we allocate whenever the realized signal $s$ is below a cutoff $t(v)$; if $\F$ has an atom at $t(v)$, we randomize only at the threshold $s=t(v)$. The cutoff and tie-breaking probability are chosen precisely so that the ex ante allocation probability  matches $x(v)$, i.e., $\E_{s\sim\F}[x(s,v)]=x(v)$. We note, when $\F$ is continuous, $\P[s=t(v)]=0$ and no randomization is needed.

Fix $v$ and let $s\sim\F$. Recall that $\F(t)=\P[s<t]$.
Choose a cutoff $t=t(v)$ such that
\[
\F(t)\le x(v)\le \F(t)+\P_{s}[s=t].
\]
Define the tie-breaking probability
\[
\alpha(v)\triangleq
\begin{cases}
\dfrac{x(v)-\F(t)}{\P_s[s=t]} & \text{if } \P_s[s=t]>0,\\
1 & \text{if } \P_s[s=t]=0.
\end{cases}
\]
and set
\begin{align*}
x(s,v) & = \begin{cases}
1 & \text{if } s<t, \\
\alpha(v) & \text{if } s=t, \\
0 & \text{if } s>t.
\end{cases}, \text{ and } \\
p(s,v) & = s \cdot x(s,v) - (1-C) \cdot s + p(v) - \Ex{s'\sim \F}{s' \cdot x(s',v) - (1-C) \cdot s'}~.
\end{align*}
By construction, $\alpha(v)\in[0,1]$ and
\[
\Ex{s\sim\F}{x(s,v)} = \Pr[s \sim \F]{s<t} + \alpha(v) \cdot \Pr[s\sim \F]{s=t} = \F(t)+\alpha(v) \cdot \Pr[s \sim \F]{s=t} = x(v)~.
\]
Moreover, taking expectations in the definition of $p(s,v)$ yields $\E_{s\sim\F}[p(s,v)]=p(v)$. 
All the constraints in problem \eqref{eq:lp-obj1} except for \eqref{eq:lp-c1} are straightforwardly satisfied. For constraint \eqref{eq:lp-c1} to hold, it suffices to verify 
\[
\forall v, \quad \Ex{s\sim\F}{s \cdot x(s,v)} - (1-C) \cdot \Ex{s\sim\F}{s} \le p(v)~.
\]
Since $(x(\cdot),p(\cdot))$ satisfies the incentive-compatibility and individual-rationality constraints in \eqref{eq:lp-obj2}, 
constraint \eqref{eq:lp-c2} is equivalent to
\[
\forall v, \quad \int_0^v x(w)  \dd w  - p(0) \le (1-C) \cdot \Ex{s\sim\F}{s} + \int_0^v \F(w) \dd w~.
\]
Let $K(v) = \int_0^v x(w) \dd  w  - p(0)$ and $H(v) = (1-C) \cdot \E_{s\sim\F}[s] + \int_0^v \F(w) \dd w$. Then $K(v)\le H(v)$ for all $v$. Consider their convex conjugates. We have that the convex conjugate of $K(\cdot)$ satisfies
\[
K^*(x(v)) = \max_{y} \left\{ x(v) \cdot y - K(y) \right\} = v \cdot x(v) - \int_0^v x(w) \dd w  + p(0)= p(v)~.
\]
Similarly, the convex conjugate of $H(v)$ is
\begin{align*}
H^*(x(v)) & = \max_y \left\{ x(v) \cdot y - H(y) \right\} \\
& = x(v) \cdot F^{-1}(x(v))- \int_0^{F^{-1}(x(v))} \F(w)\dd w -(1-C) \cdot \Ex{s \sim \F}{s} \\ 
& = \int_0^{\F^{-1}(x(v))} s\dd \F(s) -(1-C) \cdot \Ex{s\sim \F}{s} \\
& = \Ex{s \sim \F}{s \cdot x(s,v)} - (1-C) \cdot \Ex{s\sim \F}{s}.
\end{align*}
Since $K(v) \le H(v)$ for every $v$, and $x(\cdot)$ is nondecreasing under incentive compatibility, we have
\[
\forall v, \quad K^*(x(v)) \ge H^*(x(v)) \iff p(v) \ge \Ex{s\sim\F}{s \cdot x(s,v)} - (1-C) \cdot \Ex{s\sim\F}{s}~.
\]
which is exactly the condition needed for \eqref{eq:lp-c1}.
Therefore, for any feasible solution to problem \eqref{eq:lp-obj2}, our constructed solution is feasible to problem \eqref{eq:lp-obj1} and attains the same objective value. 
Hence, the two problems are equivalent.
\end{proof}

The objective \eqref{eq:lp-obj2} can be interpreted as a single-buyer, single-item revenue maximization problem, subject to the additional constraint \eqref{eq:lp-c2}. This constraint requires that the utility of a buyer with valuation $v$ is bounded above by $(1-C) \cdot \E_{s\sim\F}[s] + \E_{s\sim\F}[(v-s)^+]$. Constraint \eqref{eq:lp-c2} is the central structural feature of the reduction. In a standard single-buyer pricing problem, IC and IR only impose a lower bound on utility through convexity and monotonicity. Here, however, the accurate regime imposes an upper bound on the utility available in the hallucinatory regime: otherwise, a buyer who knows the signal is accurate would prefer to mimic the hallucinatory case. Therefore, problem \eqref{eq:lp-obj2} captures the higher-order informational asymmetry by encoding the consistency requirement in the accurate regime as a utility-cap constraint in the hallucinatory regime.

Without loss of generality, any single-buyer, single-item incentive-compatible and individually rational mechanism can be implemented as a randomized posted-price auction: the seller posts a price $t \sim \G$ with $\G: [0, v_{\max}] \rightarrow [0, 1]$, and the buyer purchases the item and pays $t$ if and only if $v \geq t$. 
Under this mechanism, a buyer with valuation $v$ secures an expected utility equal to the integral over all prices below their valuation, $\int_0^v (v-t) \dd \G(t)$. Concurrently, the seller's expected revenue is calculated over the distribution of the buyer's valuation $v \sim \F$. A trade occurs whenever $v \geq t$, yielding an expected revenue of $t \cdot \P[v \geq t]$ for any given price $t$. Thus, integrating over the seller's pricing strategy $\G$, the original problem \eqref{eq:lp-obj2} can be rewritten as the following constrained maximization problem:
\begin{subequations}
\addtocounter{equation}{-1}
    \begin{align}
\Rev(\F,C) = \max_{\G: \dd \G \ge 0}: \quad &  \int_0^{v_{\max}} t\cdot (1 - \F(t)) \dd \G(t)  \label{eq:lp-ppa-primal}\\
\text{subject to}: \quad & \int_0^v (v-t) \dd \G(t) \le (1-C) \cdot \Ex{s\sim\F}{s} + \Ex{s\sim\F}{(v-s)^+} & \forall v \in [0, v_{\max}]  \label{eq:lp-ppa-primal-cons1} \\
\quad & \int_0^{v_{\max}}  \dd \G(t) \leq 1 \label{eq:lp-ppa-primal-cons2}
\end{align}
\end{subequations}
where $s$ is a random variable drawn from $\F$. 

Problem \eqref{eq:lp-obj2} and \eqref{eq:lp-ppa-primal} characterize the optimal robust mechanism for any given consistency level $C$. However, because the resulting objective value depends on the distribution $\F$, this characterization does not yield an explicit formula for the optimal robustness guarantee. In the next subsection, we derive the exact robustness guarantee $R(C)$ achievable for each consistency $C$.
 \newcommand{\upperV}{U}

\subsection{Worst-case Robustness}\label{subsec::main-proof-2}
In this section, we characterize the optimal worst-case robustness guarantee achievable under a fixed consistency requirement $C$.
Recall that, for a fixed $C$, the robustness ratio is the worst-case revenue approximation over all priors:
\begin{align}
    R(C) = \min_{\F} \frac{\Rev(\F,C)}{\opt(\F)}. \label{eq:robust-def}
\end{align} 
In this section, we assume $C \in (0, 1)$. The endpoint cases are straightforward: when $C=0$, the consistency requirement is vacuous, so the seller can simply ignore the signal and attain $\opt(\F)$; when $C=1$, the corresponding guarantee is given by Theorem~\ref{thm:perfect}, proved in Section~\ref{sec:perfect}. We therefore focus on the interior regime.

Our proof has three steps. First, we write $\Rev(\F,C)$ in \eqref{eq:lp-ppa-primal}, as its dual minimization problem. Second, fixing the dual variables, we identify the worst-case prior $\F$. Third, we reduce the dual search space and solve the remaining low-dimensional optimization.

To analyze this worst-case ratio, we first apply duality to transform the primal maximization problem $\Rev(\F, C)$ in \eqref{eq:lp-ppa-primal} into a minimization problem.
\begin{align}
\begin{aligned}
\Rev(\F,C) = \min_{\eta \geq 0, \beta}: \quad & (1-C) \cdot \Ex{s\sim\F}{s} \cdot \beta(0) + \int_0^{v_{\max}} \beta(v) \cdot \F(v) \dd v + \eta \\
\text{subject to}: \quad & \eta + \int_t^{v_{\max}} \beta(v) \dd v \ge t \cdot (1 - \F(t)) \quad & \forall t \in [0, {v_{\max}}] \\
& \beta \text{ is decreasing and } \beta(v_{\max}) = 0.
\end{aligned}
\label{eq:lp-dual}
\end{align}

\begin{lemma}
    Problem \eqref{eq:lp-dual} is the dual program of the primal problem \eqref{eq:lp-ppa-primal}.
\end{lemma}
\begin{proof}
Let $\lambda$ be the nonnegative measure dual to constraint \eqref{eq:lp-ppa-primal-cons1}, and let $\eta$ be the dual variable corresponding to the normalization \eqref{eq:lp-ppa-primal-cons2} of the price distribution. The dual of \eqref{eq:lp-ppa-primal} is
\begin{align*}
\min_{\eta \geq 0, \lambda}: \quad & (1-C) \cdot \Ex{s \sim \F}{s} \cdot \int_0^{v_{\max}} \dd \lambda(v) + \int_0^{v_{\max}} \Ex{s \sim \F}{(v-s)^+} \dd \lambda(v) + \eta \\
\text{subject to}: \quad & \eta + \int_t^{v_{\max}} (v-t) \dd \lambda(v) \ge t \cdot (1 - \F(t)) \quad & \forall t \in [0, {v_{\max}}]\\
& \dd \lambda \ge 0.
\end{align*}
Notice that
\begin{align*}
   \quad \int_0^{v_{\max}}  \Ex{s \sim \F}{(v-s)^+} \dd \lambda(v) 
 = \int_0^{v_{\max}}  \F(v) \left(\lambda({v_{\max}})-\lambda(v) \right) \dd v 
\end{align*}
and 
$$\int_t^{v_{\max}} (v-t) \dd \lambda(v) = \int_t^{v_{\max}} \left(\lambda({v_{\max}})-\lambda(v)\right) \dd v. $$
Define $\beta(v) = \lambda({v_{\max}}) - \lambda(v)$. Then $\beta$ is decreasing, $\beta(v_{\max})=0$, and the dual becomes exactly \eqref{eq:lp-dual}.
\end{proof}

The remainder of this section is to fully characterize $R(C)$. The representation in  \eqref{eq:robust-def} shows that this reduces to jointly minimizing the following quantity over the distribution $\F$ and the dual variables $(\eta,\beta(\cdot))$:
\begin{align*}
\begin{aligned}
R(C) = \min_{\F} \min_{\eta \geq 0, \beta}: \quad & \left((1-C) \cdot \Ex{s \sim \F}{s} \cdot \beta(0) + \int_0^{v_{\max}} \beta(v) \cdot \F(v) \dd v + \eta \right) / \opt(\F)~ \\
\text{subject to}: \quad & \eta + \int_t^{v_{\max}} \beta(v) \dd v \ge t \cdot (1 - \F(t)) \quad \forall t \in [0, {v_{\max}}] \\
& \beta \text{ is decreasing and } \beta(v_{\max}) = 0.
\end{aligned}
\end{align*}

\paragraph{Characterizing the Worst-Case Distribution.}
Without loss of generality, we may assume that \(\beta(\cdot)\) is right-continuous. For any feasible decreasing \(\beta\), define
$
\tilde{\beta}(v)\triangleq \lim_{w\downarrow v}\beta(w)$.
Then \(\tilde{\beta}\) is decreasing and satisfies \(\tilde{\beta}(v)\le \beta(v)\) for all \(v\). Since \(\beta\) enters only through integrals in the constraint, replacing \(\beta\) by \(\tilde{\beta}\) preserves feasibility and weakly decreases the objective.  We first fix the pair $(\eta,\beta(\cdot))$ and ask which distribution $\F$ minimizes the objective.  Hence, we fix a targeted optimal revenue benchmark $T=\opt(\F)$, a scalar $\eta \geq 0$, and a decreasing and right continuous function $\beta(\cdot)$ such that $\beta(\upperV) = 0$ for some $U\in \R$.\footnote{Note that $\upperV$ is a general upper bound of $v$ and need not equal the absolute maximum valuation $v_{\max}$.}
We consider the following minimization problem over the distribution $\F$:
\begin{align}
\begin{aligned}
\min_{\F : \opt(\F) = T}: \quad & (1-C) \beta(0) \Ex{s\sim \F}{s} + \int_0^{\upperV} \beta(v) \F(v) \dd v + \eta  
\\
\text{subject to}: \quad & \eta + \int_t^{\upperV} \beta(v) \dd v \ge t(1 - \F(t)) \quad \forall  t \in [0, \upperV].
\end{aligned}
\label{eq:dual:U}
\end{align}

To formalize the worst-case distribution, we first define two threshold values, $v_L$ and $v_H$. These thresholds dictate where the binding constraints of the problem shift: 
\begin{align*}
v_L &= \max \left\{ v \ \Big| \ \eta + \int_v^{\upperV} \beta(w) \dd w \ge T \right\}~, \\
v_H &= \inf \left\{ v \ \mid \ \beta(v) < (1-C) \beta(0) \right\}~.
\end{align*}
The two thresholds have a simple interpretation. The quantity $B(v) \triangleq \eta + \int_v^{\upperV} \beta(w) \dd w$ is the upper bound on the revenue curve $v(1-\F(v))$ imposed by the constraint in \eqref{eq:dual:U}. 
Because $\beta(v)$ is decreasing, $B(v)$ is decreasing and convex. Roughly speaking, the dual problem \eqref{eq:dual:U} can be interpreted as searching for the cheapest convex function that dominates the revenue curve. This geometric viewpoint leads to the following argument: once such a function $B(v)$ is fixed, the worst-case distribution is obtained by making its revenue curve touch that function $B(v)$ wherever doing so lowers the objective value. 
Here $v_L$ is the last point at which this upper bound still lies 
above the benchmark $T$, so this constraint is not binding below $v_L$. The threshold $v_H$ is the first point at which $\beta(v)$ drops below $(1-C)\beta(0)$. After the objective is rewritten in the proof of \Cref{lem::feasible}, this threshold will mark the point where the coefficient of $\F(v)$ in the objective function changes sign. 

Based on this intuition, Lemma~\ref{lem:worst_case_dist} characterizes the worst-case distribution for given parameters $(\eta, \beta)$. Before that, Lemma~\ref{lem::feasible} shows that these thresholds $(v_L, v_H)$ are well defined and ordered as illustrated in Figure~\ref{fig:thresholds}, whenever \eqref{eq:dual:U} attains a value below the benchmark $T$.

\begin{figure}[htbp]
    \centering
    \includegraphics[width=0.5\textwidth]{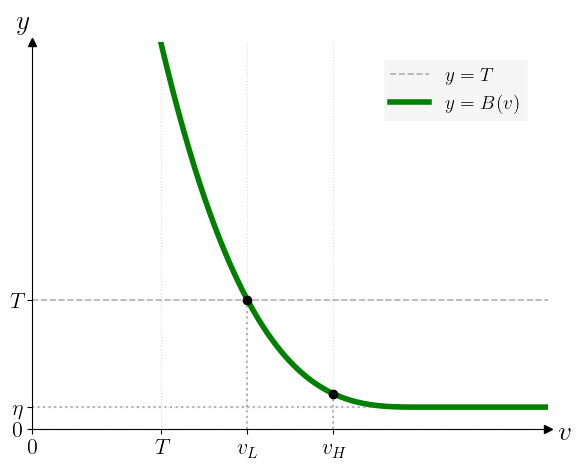}
    \caption{Illustration of the thresholds $v_L$ and $v_H$.}
    \label{fig:thresholds}
\end{figure}

\begin{lemma} \label{lem::feasible}
    Suppose \eqref{eq:dual:U} is feasible and its objective value is strictly less than $T$. Then $\beta(0) > 0$, the thresholds $v_L$ and $v_H$ are well defined, and they satisfy $T \leq v_L \leq v_H$.
\end{lemma}
\begin{proof}
Let $B(v) = \eta + \int_v^{\upperV} \beta(w) \dd w$. Because  $\beta(\cdot)$ is decreasing and $\beta(\upperV) = 0$, $B(v)$ is a continuous, non-increasing, and convex function of $v$. 

First, we establish that $v_H$ is well-defined and $\beta(0) > 0$. If $\beta(0) = 0$, then monotonicity of $\beta$ implies $\beta(v) = 0$ for all $v$. In that degenerate case the constraint in \eqref{eq:dual:U} collapses to $B(v)=\eta \ge t(1-\F(t))$ for every $t$, so feasibility requires $\eta \ge \opt(\F)=T$. The objective would then be at least $T$, contradicting the standing assumption that the objective is strictly below $T$. Hence $\beta(0)>0$. Since $\beta$ decreases from the positive value $\beta(0)$ to $0$ at $\upperV$, there must be a first point at which it drops below $(1-C)\beta(0)$; this is exactly $v_H$.

Next, we show $B(v_H)$ is a lower bound on the objective. For any valid distribution $\F$, the objective value is:
$$\begin{aligned}
    &(1-C)\beta(0)\Ex{s\sim \F}{s} + \int_0^{\upperV} \beta(v)\F(v) \dd v + \eta \\
    & \quad = (1-C)\beta(0) \int_0^{\upperV} (1-\F(v)) \dd v + \int_0^{\upperV} \beta(v)\F(v) \dd v + \eta \\
    & \quad = B(0) - \int_0^{\upperV} \left[ \beta(v) - (1-C)\beta(0) \right] (1-\F(v)) \dd v \\
    & \quad \geq B(0) - \int_0^{v_H} \left[ \beta(v) - (1-C)\beta(0) \right] \dd v \\
    & \quad \geq B(v_H).
\end{aligned}$$
The first inequality uses only that $\F(v) \in [0,1]$ and that $\beta(v) - (1-C)\beta(0) \ge 0$ on $[0,v_H)$ and $\beta(v) - (1-C)\beta(0) \le 0$ on $[v_H, U]$. The second inequality discards a nonnegative term because $\beta(0) \ge 0$. Since the objective is assumed to be strictly below $T$, we conclude that this lower bound must satisfy $B(v_H) < T$.

Finally, we show $v_L$ is well-defined and lies before $v_H$. Because $\opt(\F)=T$, the revenue curve $t(1-\F(t))$ reaches the value $T$ somewhere at $t\ge T$. The constraint in \eqref{eq:dual:U}, namely $B(t) \ge t(1-\F(t))$, therefore implies that $B$ must also stay weakly above $T$ at least until the revenue benchmark is met. On the other hand, we have just shown that $B(v_H)<T$. Since $B$ is continuous and non-increasing, it must cross the level $T$ at some point in $[T,v_H)$. By definition, the last point where $B(v)\ge T$ is precisely $v_L$, and therefore $T \le v_L <  v_H$.
\end{proof}

With the threshold structure in place, we can now determine the minimizing distribution explicitly. Figure~\ref{fig:worst-case-dist-constraints} illustrates the revenue curve constraints that govern the worst-case distribution. 

\begin{figure}[htbp]
    \centering
    \includegraphics[width=0.5\textwidth]{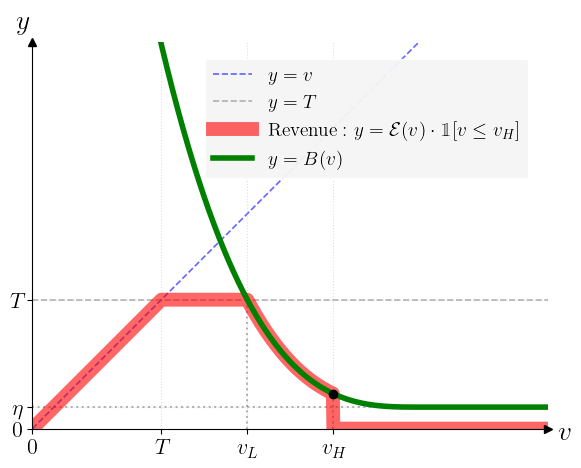}
    \caption{Revenue-curve constraints that determine the worst-case distribution. }
    \label{fig:worst-case-dist-constraints}
\end{figure}
\begin{lemma}[The Worst-Case Distribution]
\label{lem:worst_case_dist}
Fix $T$, $\eta < T$, and a decreasing function $\beta(\cdot)$ with $\beta(\upperV)=0$. If the objective in \eqref{eq:dual:U} is strictly less than $T$, then a minimizing cumulative distribution function is given by
$$
\F(v) = \begin{cases} 
    0 & v \in [0, T) \\ 
    1 - \frac{T}{v} & v \in [T, v_L) \\ 
    1 - \frac{\eta + \int_v^{\upperV} \beta(w) \dd w}{v} & v \in [v_L, v_H) \\ 
    1 & v \in [v_H, \infty) 
\end{cases}.
$$
\end{lemma}

\begin{proof}
Using the identity $\Ex{s\sim \F}{s} = \int_0^{\upperV} (1 - \F(v)) \dd v$, we rewrite the objective in terms of $\F(v)$:
$$
(1-C)\beta(0) \int_0^{\upperV} (1 - \F(v)) \dd v + \int_0^{\upperV} \beta(v) \F(v) \dd v + \eta.
$$
Collecting the terms involving $\F(v)$, the coefficient of $\F(v)$ in the integral is
$$
\beta(v) - (1-C)\beta(0).
$$
Since $\beta(\cdot)$ is decreasing and satisfies \(\beta(\upperV)=0\), this coefficient is also decreasing in $v$. By definition of $v_H$, it is nonnegative for $v < v_H$ and nonpositive for $v \ge v_H$. Hence, on $[0,v_H)$ the objective is weakly decreased by choosing $\F(v)$ as small as feasibility allows, while on $[v_H,\upperV]$ it is weakly decreased by choosing $\F(v)$ as large as feasibility allows. 

For every $v > v_H$, feasibility implies that the revenue $v\of{1-\F(v)}\le \eta + \int_v^{\upperV} \beta(w) \dd w \le B(v_H) <T $. Thus, the interval \((v_H,\upperV]\) is not needed to attain the revenue target \(T\), so setting $\F(v)=1$ for $v \ge v_H$ preserves both feasibility of \eqref{eq:dual:U} and the constraint \(\opt(\F)=T\).

For $v < v_H$, feasibility is governed by three upper bounds on the revenue curve $v(1-\F(v))$:
\begin{itemize}
    \item the trivial probability bound $v(1-\F(v)) \le v$,
    \item the normalization $\opt(\F)=T$, which implies $v(1-\F(v)) \le T$,
    \item the constraint in \eqref{eq:dual:U}, namely $v(1-\F(v)) \le \eta + \int_v^{\upperV} \beta(w) \dd w$.
\end{itemize}
Accordingly, define the envelope
$$
\mathcal{E}(v) = \min\left\{v,\, T,\, \eta + \int_v^{\upperV} \beta(w) \dd w\right\}.
$$
This envelope is the largest revenue curve compatible with all constraints. Because the objective favors smaller values of $\F(v)$ on $[0,v_H)$, the minimizing distribution saturates this envelope pointwise there.
Thus, for $v < v_H$,
$1-\F(v)=\frac{\mathcal{E}(v)}{v}$.
Evaluating the minimum on each interval gives the stated formula:
\begin{itemize}
    \item For $v \in [0,T)$, the minimum is $v$, so $1-\F(v)=1$ and hence $\F(v)=0$.
    \item For $v \in [T,v_L)$, the minimum is $T$, so $1-\F(v)=\frac{T}{v}$.
    \item For $v \in [v_L,v_H)$, the minimum is $\eta + \int_v^{\upperV} \beta(w) \dd w$, so $1-\F(v)=\frac{\eta + \int_v^{\upperV} \beta(w) \dd w}{v}$.
\end{itemize}
This proves the piecewise characterization. 
\end{proof}

\paragraph{Characterizing the Worst-Case Pair $(\eta, \beta)$.}
With the worst-case distribution established, we next simplify the pair $(\eta,\beta)$. The reduction proceeds in two steps. We first show that $\eta$ can be normalized to zero without increasing the objective. We then show that, after this normalization, it is sufficient to consider step functions for $\beta(\cdot)$.

The next lemma justifies the first reduction step. It shows that any feasible pair can be converted into another feasible pair with $\eta=0$ and no larger objective value.
\begin{lemma}[Elimination of the Constant $\eta$]
\label{lem:eliminate_mu}
Fix a revenue target $T$, a feasible pair $(\eta, \beta(\cdot))$, and an upper bound $\upperV$ in \eqref{eq:dual:U}$,$ and suppose the objective value is strictly less than $T$. Then there exists another feasible pair $(\tilde{\eta}, \tilde{\beta}(\cdot))$ with $\tilde{\eta}=0$ and a modified upper bound $\tilde{\upperV}$ whose objective value is no larger than the original one. Moreover, if
\[
\tilde{v}_H \triangleq \inf \left\{ v \mid \tilde{\beta}(v) < (1-C) \tilde{\beta}(0) \right\},
\]
then $\tilde{\beta}(\tilde{v}_H) = 0$.
\end{lemma}

\begin{proof}
Let $\F^*$ be the worst-case distribution given by Lemma~\ref{lem:worst_case_dist}, and let $v_H$ be the associated threshold. By construction, $\F^*(v)=1$ for all $v \ge v_H$.

Define
\[
\tilde{\upperV} = v_H + \frac{\eta + \int_{v_H}^{\upperV} \beta(w) \dd w}{(1-C)\beta(0)},
\qquad
\tilde{\eta}=0,
\]
and
\[
\tilde{\beta}(v)=
\begin{cases}
\beta(v) & v < v_H, \\
(1-C)\beta(0) & v \in [v_H,\tilde{\upperV}), \\
0 & v \ge \tilde{\upperV}.
\end{cases}
\]
Then $\tilde{\beta}$ is decreasing, $\tilde{\beta}(\tilde{\upperV})=0$, and the first point at which $\tilde{\beta}(v)$ falls below $(1-C)\tilde{\beta}(0)$ is $\tilde{v}_H=\tilde{\upperV}$. In particular, $\tilde{\beta}(\tilde{v}_H)=0$.

For any $t < v_H$,
\begin{align*}
\int_t^{\tilde{\upperV}} \tilde{\beta}(v) \dd v
&= \int_t^{v_H} \beta(v) \dd v + \eta + \int_{v_H}^{\upperV} \beta(v) \dd v \\
&= \eta + \int_t^{\upperV} \beta(v) \dd v \\
&\ge t \bigl(1-\F^*(t)\bigr).
\end{align*}
For $t \ge v_H$, the right-hand side is zero because $\F^*(t)=1$. Hence $(\tilde{\eta},\tilde{\beta})$ is feasible for the same worst-case distribution.

Finally,
\begin{align*}
& (1-C)\beta(0)\Ex{s\sim \F^*}{s} + \int_0^{\upperV} \beta(v) \F^*(v) \dd v + \eta  = (1-C)\tilde{\beta}(0)\Ex{s\sim \F^*}{s} + \int_0^{\tilde{\upperV}} \tilde{\beta}(v) \F^*(v) \dd v.
\end{align*}
Thus, the objective value is preserved and, in particular, does not increase.
\end{proof}
\begin{figure}[!t]
    \centering    \includegraphics[width=0.5\textwidth]{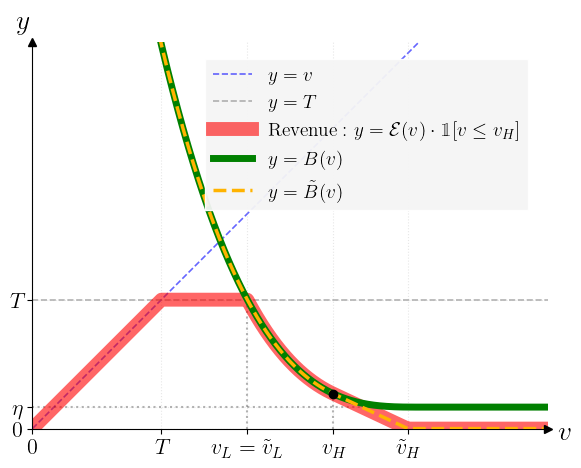}
    \caption{Illustration of the first reduction steps for $(\eta,\beta)$. For a given pair $(\eta,\beta)$, which induces $B(v) = \int_{v}^{\upperV} \beta(w) \dd w + \eta$, the figure shows how $\eta$ can be eliminated and $\beta$ can be transformed into $\tilde{\beta}$, which induces $\tilde{B}(v)$, without increasing the objective. The red curve is the revenue curve of the worst-case distribution under the pair $(0, \tilde{\beta})$.}
    \label{fig:worst-case-dist}
\end{figure}
In \Cref{fig:worst-case-dist}, we illustrate the original $B(v) = \eta + \int_v^{\upperV} \beta(w) \dd w$ in green solid line and the modified $\tilde{B}(v) = \int_v^{\tilde{\upperV}} \tilde{\beta}(w) \dd w$ in yellow dashed line.

The next lemma justifies the second reduction step. Once $\eta=0$, it shows that one may further restrict attention to step functions without increasing the objective value.
\begin{lemma}[Optimality of Step-Function $\beta$]
\label{lem:step_beta}
Suppose $\eta=0$, $\beta(v_H)=0$, and the objective in \eqref{eq:dual:U} is strictly less than $T$. 
Then there exists a step function \(\hat{\beta}(\cdot)\), equal to \(\beta(v_L)\) up to a threshold and zero thereafter, such that when \(\hat{\beta}\) is paired with its induced worst-case distribution \(\hat{\F}\), the resulting objective value is weakly smaller than that under \((\beta,\F)\).
\end{lemma}

\begin{proof}
Since $\eta = 0$, let
\[
B(v) = \int_v^{\upperV} \beta(w) \dd w.
\]
Applying integration by parts, and using $\F(0)=0$ and $B(\upperV)=0$, we obtain
$$
\int_0^{\upperV} \beta(v) \F(v) \dd v = \int_0^{\upperV} B(v) \dd \F(v).
$$
Hence the objective becomes
\begin{align}
(1-C) \beta(0) \Ex{s\sim \F}{s} + \int_0^{\upperV} B(v) \dd \F(v). \label{obj:rewrite}
\end{align}
We now construct a step function $\hat{\beta}(\cdot)$:
$$
\hat{\beta}(v) = \begin{cases}
\beta(v_L) & \text{for } v < v_L + \frac{T}{\beta(v_L)}, \\
0 & \text{for } v \ge v_L + \frac{T}{\beta(v_L)}.
\end{cases}
$$
With this construction,
\[
\int_{v_L}^{\upperV} \hat{\beta}(w) \dd w = \int_{v_L}^{\upperV} \beta(w) \dd w = T,
\]
and $\hat{\beta}(v_L)=\beta(v_L)$. Let
\[
\hat{B}(v)=\int_v^{\upperV} \hat{\beta}(w)\dd w.
\]
Then $\hat{B}(v) \le B(v)$ for all $v$, while $\hat{B}(v_L)=B(v_L)=T$. Both $B$ and $\hat{B}$ are decreasing and convex.

Let $\hat{\F}$ be the worst-case distribution generated by this new $\hat{\beta}(\cdot)$. By construction, $\hat{v}_L=v_L$, where
\[
\hat{v}_L \triangleq \max \left\{ v \ \Big| \ \int_v^{\upperV} \hat{\beta}(w) \dd w \ge T \right\},
\]
and $\hat{v}_H \le v_H$, where
\[
\hat{v}_H \triangleq \inf \left\{ v \mid \hat{\beta}(v) < (1-C)\hat{\beta}(0) \right\} = v_L + \frac{T}{\beta(v_L)}.
\]

We claim that replacing $(\beta,\F)$ by $(\hat{\beta},\hat{\F})$ weakly decreases the objective in \eqref{obj:rewrite}. The comparison is transparent once the objective is split into its two terms.

First, because $\beta(\cdot)$ is decreasing,
\[
\hat{\beta}(0)=\beta(v_L) \le \beta(0),
\]
so the coefficient of $\E[s]$ weakly decreases.

Second, the induced distribution also shifts mass weakly downward. By Lemma~\ref{lem:worst_case_dist}, the two distributions coincide on $[0,v_L)$, while for $v \in [v_L,v_H)$ we have
\[
\F(v) = 1 - \frac{B(v)}{v},
\qquad
\hat{\F}(v) = 1 - \frac{\hat{B}(v)}{v}.
\]
Since $\hat{B}(v) \le B(v)$, it follows that $\hat{\F}(v) \ge \F(v)$ on that interval. Therefore $\E_{\hat{\F}}[s] \le \E_{\F}[s]$.

It remains to compare the integral term. On $[0,v_L)$, the measures $\dd \hat{\F}$ and $\dd \F$ coincide, and $\hat{B}(v) \le B(v)$, so
\[
\int_0^{v_L} \hat{B}(v) \dd \hat{\F}(v) \le \int_0^{v_L} B(v) \dd \F(v).
\]
On $[v_L,\upperV]$, write the integral in quantile form. Let $q=1-\F(v)$ and let $v(q)$ denote the corresponding inverse quantile. Then
\[
\int_{v_L}^{\upperV} B(v) \dd \F(v) = \int_0^{1-\F(v_L)} B(v(q)) \dd q.
\]
For $v \ge v_L$, Lemma~\ref{lem:worst_case_dist} gives $1-\F(v)=B(v)/v$, so in quantile form $B(v(q))=q\,v(q)$. Hence,
\[
\int_{v_L}^{\upperV} B(v) \dd \F(v) = \int_0^{1-\F(v_L)} q\,v(q) \dd q.
\]
Applying the same change of variables to $(\hat{B},\hat{\F})$ yields
\[
\int_{v_L}^{\upperV} \hat{B}(v) \dd \hat{\F}(v) = \int_0^{1-\hat{\F}(v_L)} q\,\hat{v}(q) \dd q.
\]
Because $\hat{v}_L=v_L$ and the two distributions agree at $v_L$, the integration bounds $1-\F(v_L) = 1-\hat{\F}(v_L)$. Moreover, $\hat{\F}(v) \ge \F(v)$ implies $\hat{v}(q) \le v(q)$ for every $q$. Therefore,
\[
\int_0^{1-\hat{\F}(v_L)} q\,\hat{v}(q) \dd q \le \int_0^{1-\F(v_L)} q\,v(q) \dd q.
\]
Combining the two intervals shows that the integral term weakly decreases as well.

Since both terms in \eqref{obj:rewrite} weakly decrease under the replacement, any optimal solution can be represented by a step function.
\end{proof}

\Cref{fig:step-beta} illustrates the reduction from a general decreasing function $\beta$ to a step function $\hat{\beta}$, together with the induced revenue curve under $\hat{\beta}$.

\begin{figure}[!t]
    \centering
    \includegraphics[width=0.5\textwidth]{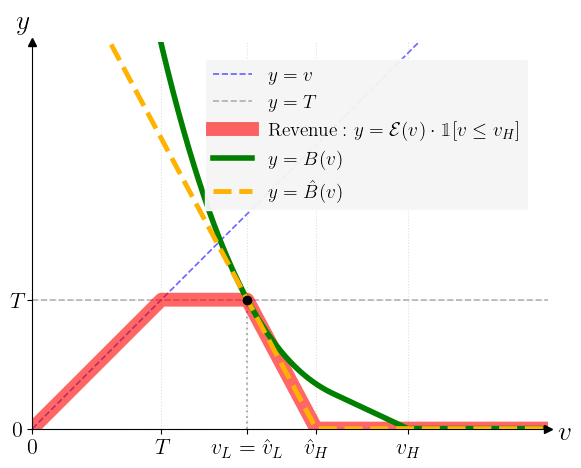}
    \caption{Step-function reduction for $\beta(\cdot)$.}
    \label{fig:step-beta}
\end{figure}

Lemma~\ref{lem:worst_case_dist} -- Lemma~\ref {lem:step_beta} reduce the original infinite-dimensional worst-case optimization over $(\eta,\beta(\cdot), \F)$ to a two-parameter family.
After these two reduction steps, the remaining optimization is low-dimensional and can be solved explicitly. We next evaluate this reduced problem and thereby complete the proof of Theorem~\ref{thm:frontier}.

\begin{proof}[Proof of Theorem~\ref{thm:frontier}]
Given a fixed $T$, we consider the ratio between the objective of the following optimization and $T$.
\begin{align*}
\begin{aligned}
\min_{\eta \geq 0, \beta, \F: \opt(\F) = T}: \quad & (1-C) \cdot \E[s] \cdot \beta(0) + \int_0^{v_{\max}} \beta(v) \cdot \F(v) \dd v + \eta \\
\text{subject to}: \quad & \eta + \int_t^{v_{\max}} \beta(v) \dd v \ge t \cdot (1 - \F(t)) \quad \forall t \in [0, {v_{\max}}] \\
& \beta \text{ is decreasing and } \beta(v_{\max}) = 0.
\end{aligned}
\end{align*}

We first show that this ratio is always at least $R^{\star}(C)$, thereby completing the proof of achievability.

By Lemma~\ref{lem:eliminate_mu} and Lemma~\ref{lem:step_beta}, it suffices to study the reduced class of pairs $(\eta,\beta(\cdot))$ with $\eta=0$ and step-function $\beta(\cdot)$. Moreover, Lemma~\ref{lem::feasible} implies $v_L \ge T$, which implies
\[
\int_T^{\infty} \beta(v) \dd v \ge T.
\]
Therefore, we may parameterize $\beta(\cdot)$ by the optimal revenue target $T$, $v_L \ge T$, and a constant $\beta > 0$:
$$
\beta(v) = \begin{cases}
\beta &\text{for } v < v_L + \frac{T}{\beta}, \\
0 &\text{otherwise}.
\end{cases}
$$
Let $\upperV = v_L + \frac{T}{\beta}$.
Under this, applying Lemma~\ref{lem:worst_case_dist}, the worst-case distribution takes the form
$$
\F(v) = \begin{cases}
0 &\text{for } v < T, \\
1 - \frac{T}{v} &\text{for } v \in [T, v_L), \\
1 - \frac{T - \beta(v - v_L)}{v} &\text{for } v \in \left[v_L, v_L + \frac{T}{\beta}\right), \\
1 &\text{for } v \geq v_L + \frac{T}{\beta}.
\end{cases}
$$

We then evaluate the objective:
\begin{align*}
 &(1-C) \beta(0) \E[s] + \int_0^{\upperV} \beta(v) \F(v) \dd v + \eta \\
 &\quad = (1 - C) \cdot \beta\left[T + \int_{T}^{v_L} \frac{T}{v} \dd v + \int_{v_L}^{v_L + \frac{T}{\beta}} \frac{T - \beta(v - v_L)}{v} \dd v \right] \\
 &\quad \quad \quad + \int_{T}^{v_L} \frac{T}{v^2} (T + \beta(v_L - v)) \dd v + \int_{v_L}^{v_L + \frac{T}{\beta}} \frac{T + \beta v_L}{v^2} (T - (v - v_L) \beta) \dd v \\
 &\quad = T + \beta v_L - C \beta \left[ T \ln\left(\frac{v_L}{T}\right) + (T + \beta v_L) \ln\left(1 + \frac{T}{\beta v_L}\right) \right].
\end{align*}
This expression is convex in $v_L$. Its unconstrained minimizer is
\[
v_L = \frac{T}{\beta \left( e^{\frac{1}{C \beta}} - 1 \right)}.
\]
Because this value is strictly smaller than $T$ while feasibility requires $v_L \ge T$, the constrained minimum is attained at $v_L=T$.

Substituting $v_L = T$ further simplifies the minimum objective value to
$$
T(1 + \beta) - T \cdot C (\beta + \beta^2) \ln\left(1 + \frac{1}{\beta}\right).
$$
Since $\opt(\F)=T$, dividing by $T$ yields the lower bound
\[
R(C) \ge \inf_{\beta>0} \left[(1+\beta) - C(\beta+\beta^2) \ln\left(1+\frac{1}{\beta}\right)\right].
\]

To prove optimality, fix any $\beta>0$ and any $T>0$, and define $v_H=T+T/\beta$. Choose $T$ small enough so that $v_H \le v_{\max}$. Consider the distribution
\[
\F_\beta(v)=
\begin{cases}
0 & v \in [0,T), \\
1-\dfrac{T-\beta(v-T)}{v} & v \in [T,v_H), \\
1 & v \ge v_H.
\end{cases}
\]
\begin{figure}[!t]
    \centering
    \includegraphics[width=0.5\textwidth]{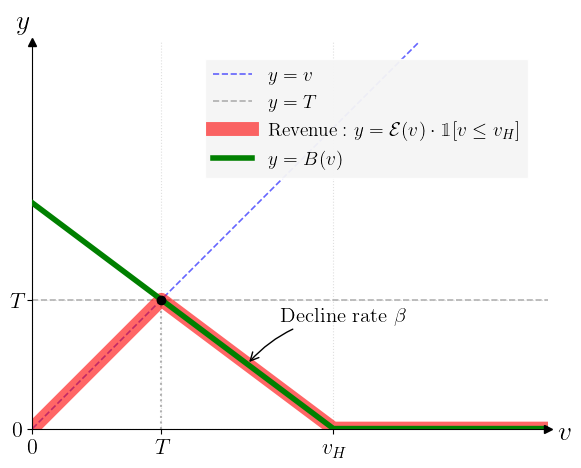}
    \caption{Revenue curve for the worst-case distribution $\F_\beta$.}
    \label{fig:worst-case-d}
\end{figure}
Its revenue curve (illustrated in Figure~\ref{fig:worst-case-d}) is
\[
v\bigl(1-\F_\beta(v)\bigr)=
\begin{cases}
v & v \in [0,T), \\
T-\beta(v-T) & v \in [T,v_H), \\
0 & v \ge v_H,
\end{cases}
\]
so $\opt(\F_\beta)=T$. Taking the feasible pair
\[
\eta=0,
\qquad
\beta(v)=
\begin{cases}
\beta & v < v_H, \\
0 & v \ge v_H,
\end{cases}
\]
the objective is exactly
\[
T \left[(1+\beta) - C(\beta+\beta^2) \ln\left(1+\frac{1}{\beta}\right)\right].
\]
By weak duality,
\[
\frac{\Rev(\F_\beta,C)}{\opt(\F_\beta)}
\le (1+\beta) - C(\beta+\beta^2) \ln\left(1+\frac{1}{\beta}\right).
\]
Taking the infimum over $\beta>0$ proves the reverse inequality. Hence, 
\[
R^{\star}(C) \;=\; \inf_{\beta>0}\Bigl[(1+\beta) - C\,(\beta+\beta^2)\,\ln\bigl(1+\tfrac{1}{\beta}\bigr)\Bigr].
\]
\end{proof} \section{Conclusion}
This paper studies pricing when the seller possesses side information of uncertain reliability. By adopting a consistency-robustness framework, we characterize the exact tradeoff between exploiting a signal as if it were accurate and protecting against the possibility that it is hallucinatory. We demonstrate that even when the seller cannot verify the signal's reliability, they can still design mechanisms that leverage the buyer’s awareness of the information reliability.

Our work suggests a design perspective that moves beyond the classical predict-then-optimize paradigm. Traditional mechanism design relies primarily on information directly held by the designer. By contrast, our framework shows that decisions can be built around higher-order knowledge, specifically, what the designer knows others know. This shifts the role of information from a passive constraint to an active strategic resource. Consequently, samples need not be used solely to estimate distributions; they can be integrated directly into the mechanism's rules to support robust and consistent decisions.

We conclude by highlighting several promising directions for future research.
\begin{itemize}
\item 
A natural extension is to relax the assumption that the signal is either perfectly accurate or purely hallucinatory, and instead allow arbitrary mixtures of the two cases. Our results suggest that the seller may still be able to exploit the buyer’s knowledge of the signal distribution. Such a generalization would help bridge our formulation with the Bayesian model of \cite{lobel2025}.

\item 
Another natural direction is to study the multi-bidder setting. Beyond leveraging each buyer’s knowledge of her own signal, one can ask whether a buyer may also have better information about her competitors than the seller does. Our results suggest that such higher-order informational asymmetries, together with the resulting strategic interactions, may themselves become a fundamental resource in mechanism design.
\end{itemize}

Beyond mechanism design, our work also contributes to the literature on algorithms with predictions by introducing a simple prediction paradigm in a stochastic/Bayesian optimization setting.
\begin{itemize}
\item
The algorithms-with-predictions literature was initially motivated by predictions as an alternative to stochastic modeling for going beyond worst-case analysis (see, e.g., Chapter 30 of~\cite{roughgarden2021beyond}), rather than as a complement to it. Recent work on prophet inequalities~\cite{arxiv/BaiHLL25,arxiv/KehneK25}, for example, proposes predictions of the prior as a way to interpolate between stochastic and adversarial settings: consistency corresponds to an accurate prior, recovering the stochastic model, while robustness corresponds to uninformative advice, recovering the adversarial model.

By contrast, our notion of consistency and robustness is built on top of an already stochastic environment, where the prediction takes the form of additional side information. In the prophet inequality setting, the analogue of our model would be to equip the algorithm designer with a signal about the largest realized value. Consistency would then correspond to the case in which this signal is correct, while robustness would correspond to the case in which it is an independent signal unrelated to the realized instance. This appears to be an interesting question in its own right, even separate from the higher-order informational aspects emphasized in our paper.
\end{itemize}

\bibliographystyle{alpha}
\bibliography{ref}

\end{document}